\documentclass[11pt]{llncs}

\usepackage{graphicx}

\usepackage{graphicx,amssymb,amsmath}
\usepackage{graphicx}
\usepackage{graphics}
\usepackage{epsfig}
\usepackage{epstopdf}
\usepackage{amsmath}
\usepackage{latexsym}
\usepackage{subfig}
\usepackage{wrapfig}
\usepackage{multirow}
\usepackage{amsfonts}
\usepackage{cite}
\usepackage{amssymb}
\usepackage{caption}
\usepackage{algorithm}
\usepackage{algorithmic}
\usepackage[a4paper,margin=1.2in,footskip=0.25in]{geometry}

\makeatletter
\newcommand{\rmnum}[1]{\romannumeral #1}
\newcommand{\Rmnum}[1]{\expandafter\@slowromancap\romannumeral #1@}
\makeatother

\begin{document}

\title{Faster Balanced Clusterings in High Dimension\thanks{Part of this work has appeared in~\cite{DBLP:conf/cccg/Ding17} that only considered the balanced $k$-center clustering. In this paper, we significantly simplify the proofs and generalize the idea to an effective framework for three different balanced clustering problems in high dimension. 
}
}
%
%

\author{Hu Ding}
\institute{
 Department of Computer Science and Engineering\\
Michigan State University\\
School of Computer Science and Technology\\
University of Science and Technology of China\\
 \email{\tt huding@msu.edu, huding@ustc.edu.cn}
 }

%
%


\maketitle
\thispagestyle{plain}
\begin{abstract}
The problem of constrained clustering has attracted significant attention in the past decades. In this paper, we study the balanced $k$-center, $k$-median, and $k$-means clustering problems where the size of each cluster is constrained by the given lower and upper bounds. The problems are motivated by the applications in processing large-scale data in high dimension. Existing methods often need to compute complicated matchings (or min cost flows) to satisfy the balance constraint, and thus suffer from high complexities especially in high dimension. 
We develop an effective framework for the three balanced clustering problems to address this issue, and our method is based on a novel spatial partition idea in geometry.
For the balanced $k$-center clustering, we provide a $4$-approximation algorithm that improves the existing approximation factors; for the balanced $k$-median and $k$-means clusterings, our algorithms yield constant and $(1+\epsilon)$-approximation factors with any $\epsilon>0$. More importantly, our algorithms achieve linear or nearly linear running times when $k$ is a constant, and significantly improve the existing ones.
Our results can be easily extended to metric balanced clusterings and the running times are sub-linear in terms of the complexity of $n$-point metric.\\


\end{abstract}

 \section{Introduction}
\label{sec-set}
Geometric clustering is a fundamental topic in computer science and has numerous applications in the areas of data mining, data management, and machine learning~\cite{jain2010data}. Given a set of points in Euclidean space and a positive integer $k$, the problem of clustering is to partition the set into $k$ clusters to minimize some given objective function. $k$-center, $k$-median, and $k$-means clusterings are among the most widely studied clustering problems. Roughly speaking, the first one aims to cover the clusters by $k$ balls such that  the maximum radius of the balls is minimized, while the other two minimize the average (squared) distance from each input point to its nearest cluster center. Note that the cluster centers can locate arbitrarily in the Euclidean space and do not necessarily come from the input set. 
Different to geometric clustering, metric clustering considers the case that all the given points (vertices) form a metric graph and the cluster centers are chosen from the vertices; the problem is motivated by the applications in the area of facility location, e.g., how to place facilities to minimize transportation costs~\cite{manne1964plant,kuehn1963heuristic}.

In this paper, we consider the balanced $k$-center, $k$-median, and $k$-means clusterings where the size of each cluster is constrained within some given interval. 
Besides the well studied applications in data analysis and facility location, the balanced clustering problem is particularly motivated by the arising challenges in {\em big data}~\cite{DLP,ABM,BBL}. For example, if the data scale is extremely large, we need to dispatch data to multiple machines to process; at the same time we have to consider the balancedness, because the machines receiving too much data could be the bottleneck of the system and the ones receiving too little data is not sufficiently energy-efficient. Another motivation of the balanced clustering is from machine learning.  It was observed that many classification methods have the property ``locally simple but globally complex''; that is, a classification rule may be complicated on the global data, but often can be very simple in each local region~\cite{vapnik1993local}. For example, Dick et al.~\cite{DLP} proposed a new paradigm for performing data-dependent dispatching that takes advantage of such structure by sending similar data points to the same machines.

Below, we provide the formal definitions of the problems studied in this paper.

\begin{definition}
\label{def-problem}
Let $k$ be a positive integer and $P$ be a set of $n$ points in $\mathbb{R}^d$. Given $L\leq U \in \mathbb{Z}^+$, the balanced clustering is to partition $P$ into $k$ clusters $P_1, \cdots, P_k$, where each cluster $P_j$ has a cluster center $c_j\in\mathbb{R}^d$ and the size within $[L, U]$, such that the following objective function is minimized:
\begin{itemize}
\item $k$-center: $\max_{1\leq j\leq k}\max_{p\in P_j}||p-c_j||$;
\item $k$-median: $\sum^k_{j=1}\sum_{p\in P_j}||p-c_j||$;
\item $k$-means: $\sum^k_{j=1}\sum_{p\in P_j}||p-c_j||^2$.
\end{itemize}
For simplicity, we denote the problems as $k$-BCenter, $k$-BMedian, and $k$-BMeans, respectively. 
\end{definition}

Similarly, we have the definitions for the balanced clusterings in an abstract metric space.

\begin{definition}
\label{def-problem2}
Let $k$ be a positive integer and $P$ be a set of $n$ vertices in some abstract metric space. Given $L\leq U \in \mathbb{Z}^+$, the balanced clustering is to partition $P$ into $k$ clusters $P_1, \cdots, P_k$, where each cluster $P_j$ has a cluster center $c_j\in P_j$ and the size within $[L, U]$, such that the following objective function is minimized:
\begin{itemize}
\item $k$-center: $\max_{1\leq j\leq k}\max_{p\in P_j}||p-c_j||$;
\item $k$-median: $\sum^k_{j=1}\sum_{p\in P_j}||p-c_j||$;
\item $k$-means: $\sum^k_{j=1}\sum_{p\in P_j}||p-c_j||^2$.
\end{itemize}
For simplicity, we denote the problems as Metric $k$-BCenter, Metric $k$-BMedian, and Metric $k$-BMeans, respectively. 
\end{definition}

\begin{remark}
To ensure that a feasible solution exists, we assume $1\leq L\leq \lfloor\frac{n}{k}\rfloor \leq \lceil\frac{n}{k}\rceil\leq U\leq n$ in Definition~\ref{def-problem} and \ref{def-problem2}.
\end{remark}

\begin{definition}
\label{def-app}
$\forall c\geq 1$, any feasible clustering solution (of the problems in Definition~\ref{def-problem} and \ref{def-problem2}) achieving at most $c$ times the minimum value of the objective value is called a $c$-approximation.
\end{definition}



\subsection{The Related Work and Our Main Result}
\label{sec-our}

The optimal approximation results of the ordinary $k$-center clustering appeared in the 80's: Gonazlez~\cite{G85} and Hochbaum and Shmoys~\cite{HS85} respectively provided a 2-approximation and proved that any approximation ratio $c<2$ would imply $P=NP$. If $k$ is fixed, we have the {\em PTAS (Polynomial Time Approximation Scheme)} for $k$-center clustering through constructing the core-set of minimum enclosing ball~\cite{BHI,badoiu2003smaller}.

For the ordinary $k$-median and $k$-means clusterings, Arya et al.~\cite{arya2004local} and Kanungo et al.~\cite{kanungo2004local} separately gave a $(3+\epsilon)$-approximation and a $(9+\epsilon)$-approximation algorithms, where $\epsilon$ can be any small positive number. Recently, Li and Svensson~\cite{li2016approximating} proposed the concept ``pseudo-approximation'' that allows to output $k+O(1)$ clusters for $k$-median clustering; their method achieves an approximation ratio of $1+\sqrt{3}+\epsilon$, and then Byrka et al.~\cite{byrka2017improved} improved the ratio to be $2.675+\epsilon$. Both $k$-median clustering~\cite{megiddo1984complexity} and $k$-means clustering~\cite{D08,mahajan2012planar,vattanihardness} have been shown to be NP-hard. Also, it is NP-hard to achieve their approximations within some factor larger than $1$~\cite{DBLP:conf/compgeom/AwasthiCKS15,guruswami2003embeddings}. 
When the dimension $d$ is fixed, Arora et al.~\cite{arora1998approximation} gave a PTAS for $k$-median clustering, where the running time was further improved to be nearly linear by Kolliopoulos and Rao~\cite{kolliopoulos2007nearly}; recently, Cohen-Addad et al.~\cite{cohen2016local} and Friggstad et al.~\cite{friggstad2016local} provided the PTAS for low dimensional $k$-means clustering by using the local search technique. When $k$ is fixed, Kumar et al.~\cite{kumar2010linear} and Jaiswal et al.~\cite{jaiswal2014simple,jaiswal2015improved} gave the PTAS for $k$-median and $k$-means clusterings in high dimension. Moreover, Chen~\cite{DBLP:journals/siamcomp/Chen09} and further Feldman and Langberg~\cite{feldman2011unified} showed that the core-set technique is able to reduce the running times of the algorithms. We refer the reader to the recent survey on clustering problems for more details~\cite{jaiswal2018approximate}.

Several variants of $k$-center clustering given the upper~\cite{BKP,KS00,CHK,ABC,KC14} or lower~\cite{APF,EHR,AS16,DBLP:conf/waoa/AhmadianS12} bounds on cluster sizes have been extensively studied in recent years. 
In particular, Ding et al.~\cite{DHH} studied the $k$-center clustering with two-sided bounds. Further, R{\"{o}}sner and Schmidt~\cite{DBLP:conf/icalp/Rosner018} investigate different types of balanced $k$-center clustering in terms of privacy preserving. 
Most of existing methods model the problems as linear integer programmings and design novel rounding algorithms to obtain some constant factor approximations. For example, the algorithm of~\cite{DHH} yields a $6$-approximation for Metric $k$-BCenter (if using triangle inequality, their result directly implies a $12$-approximation for $k$-BCenter in Euclidean space). $k$-BMedian and $k$-BMeans are much more challenging. For the case with only the upper bound, Li~\cite{li2016approximating,li2017uniform} showed that an $O(\frac{1}{\epsilon^2}\log\frac{1}{\epsilon})$-approximation can be obtained if it allows to output $(1+\epsilon)k$ clusters. Dick et al.~\cite{DLP} applied the techniques of linear programming and min cost flow to produce several constant factor approximations for $k$-BCenter, $k$-BMedian, and $k$-BMeans, but the upper bounds are violated by some constant factors. Borgwardt et al.~\cite{borgwardt2017lp} considered the convergence of using the heuristic Lloyd's $k$-means method for $k$-BMeans. Several other balanced clustering problems have been studied as well, such as~\cite{malinen2014balanced,banerjee2006scalable}.

In fact, balanced clustering falls under the umbrella of general constrained clustering problem. Ding and Xu~\cite{DBLP:conf/soda/DingX15} provided a unified framework for solving a class of constrained $k$-median and $k$-means clusterings. Roughly speaking, given an $\epsilon>0$, the method generates a set of $O\big(2^{\text{poly}(k/\epsilon)}(\log n)^k\big)$ $k$-tuples where at least one of them yields $(1+\epsilon)$-approximation. However, to select the qualified candidate, we need to design different selection algorithms depending on different constraints (e.g., our balanced requirement). Further, Bhattacharya et al.~\cite{bhattacharya2018faster} improved the result of~\cite{DBLP:conf/soda/DingX15} with respect to both the candidate set size and running time.

\textbf{Our main result.} 
Motivated by the applications from big data, we assume that both the number of input points $n$ and dimension $d$ are large, and propose faster algorithms for $k$-BCenter, $k$-BMedian, $k$-BMeans, and their metric counterparts. In addition, we assume that the number of clusters $k$ is a constant. Actually, $k$ is usually a small number in practice (e.g., the data is distributed over less than $10$ machines). Moreover, the existing research on large-scale clustering problems often assume that either $k$ or $d$ is a constant, due to their hardness results mentioned above~\cite{G85,HS85,megiddo1984complexity,D08,mahajan2012planar,vattanihardness,DBLP:conf/compgeom/AwasthiCKS15,guruswami2003embeddings}.


Clustering problems commonly involve two key steps: \textbf{(\rmnum{1}).} determine the $k$ cluster centers and \textbf{(\rmnum{2}).} partition the input data into $k$ clusters. For ordinary clustering problems, step (\rmnum{2}) is trivial. Namely, we just need to assign each data point to its nearest cluster center to minimize the total clustering cost. However, step (\rmnum{2}) can be complicated for constrained clustering problems. In general, we need to compute the matching between input points $P$ and the obtained cluster centers, such that the objective value is minimized and the constraint (e.g., the balancedness) can be satisfied simultaneously. For example, the algorithms in~\cite{DLP,DBLP:conf/soda/DingX15,bhattacharya2018faster} all need to reduce it to be a max flow or min cost flow problem. They build a bipartite graph between $P$ and the $k$ cluster centers, and the numbers of the vertices $V$ and edges $E$ are both linear on $n$; therefore, their running times will be $\Omega(|V|\cdot|E|)$ that is at least quadratic on the input size. The authors of~\cite{DBLP:conf/soda/DingX15} proposed the open problem: can we avoid computing the high complexity matching in step (\rmnum{2})?

In this paper, we answer their question in the affirmative. Specifically, we provide a novel method to complete step (\rmnum{2}) in linear time. For $k$-BCenter, we apply a spatial partition idea to build a system of linear equations and inequalities that has the size independent of $n$, and round the feasible solution to be an integral solution (i.e., a feasible partition on $P$) efficiently without increasing the objective value. Comparing with the existing method for $k$-BCenter~\cite{DHH}, we improve the approximation ratio from $12$ to $4$ and significantly reduce the running time by avoiding to solve the large-scale matching problem.  To solve $k$-BMedian and $k$-BMeans, we generalize the spatial partition idea and replace the system of linear equations and inequalities by a model of linear programming. More importantly, a feasible partition on $P$ can be efficiently obtained as well. We adopt the results from~\cite{DBLP:conf/soda/DingX15,bhattacharya2018faster} and obtain the same constant factor and $(1+\epsilon)$-approximations, while our running times are much lower.


Furthermore, our method can be easily extended to Metric $k$-BCenter, $k$-BMedian, and $k$-BMeans with similar running times.


The rest of the paper is organized as follows. Assuming the $k$ cluster centers are given, we propose our framework to efficiently compute the balanced partition in Section~\ref{sec-partition}. Using this framework, we present the approximation algorithms for $k$-BCenter, $k$-BMedian, $k$-BMeans, and their metric counterparts in Section~\ref{sec-high}.

%
\section{Computing The Best Balanced Partition}
\label{sec-partition}

Suppose the $k$ cluster centers are fixed, we consider the problem that how to compute a balanced partition on $P$ to minimize the clustering costs of $k$-BCenter, $k$-BMedian, and $k$-BMeans. We introduce our method for $k$-BCenter first, and then extend it to the more complicated $k$-BMedian and $k$-BMeans.

\subsection{Balanced Partition for $k$-BCenter}
\label{sec-feasible}
%
%

Denote by $\{c_1, \cdots, c_k\}$ the fixed $k$ cluster centers. Since $k$-BCenter is to find $k$ balls, we need to determine the radius first. 
Given a number $r>0$, we draw $k$ balls with the radius $r$ and centered at $\{c_1, \cdots, c_k\}$ respectively. We denote the $k$ balls as $\mathcal{B}_1, \cdots, \mathcal{B}_k$.  If we can find a balanced partition of $P$, say $P_1, \cdots, P_k$, that each $P_j$ is covered by an individual ball, we say that $r$ is a {\em feasible radius}. It is easy to know that the radius must come from the $kn$ distances $\{||p-c_j||\mid p\in P, 1\leq j\leq k\}$. If we have an \textbf{oracle} to check the feasibility of each candidate $r$, we can apply binary search to find the smallest feasible radius.

A straightforward way to check the feasibility is building a bipartite graph between the $n$ points of $P$ and the $k$ balls, where a point is connected to a ball if it is covered by the ball; each ball has a capacity $U$ and demand $L$, and the maximum flow from the points to balls is $n$ if and only if the radius is feasible. The existing maximum flow algorithms, such as Ford-Fulkerson algorithm and the recent Orlin's algorithm~\cite{O13}, all take  $\Omega(n^2)$ time. If $k$ is a constant, we show that the problem can be reduced to \textbf{a system of linear equations and inequalities (SoL)} with the size independent of $n$. 

The region $\cup^k_{j=1}\mathcal{B}_j$ divides the space into $2^k-1$ parts (we ignore the region outside the union of the balls, since no point locates there; otherwise, we can simply reject this candidate $r$).
Suppose we have $t$ indices $1\leq j_1<j_2<\cdots<j_t\leq k$ with $1\leq t\leq k$. Denote by $\mathcal{R}_{(j_1, j_2, \cdots, j_t)}$ the region 
\begin{eqnarray}
(\mathcal{B}_{j_1}\cap \cdots\cap \mathcal{B}_{j_t})\setminus(\cup_{j\notin\{j_1, \cdots, j_t\}}\mathcal{B}_j).
\end{eqnarray}
We calculate the total number of points covered by $\mathcal{R}_{(j_1, j_2, \cdots, j_t)}$ and denote it as $n_{(j_1, j_2, \cdots, j_t)}$. Moreover, we assign $t$ non-negative variables 
\begin{eqnarray}
x^{j_1}_{(j_1, j_2, \cdots, j_t)},\hspace{0.1in} x^{j_2}_{(j_1, j_2, \cdots, j_t)},\hspace{0.1in}\cdots,\hspace{0.1in} x^{j_t}_{(j_1, j_2, \cdots, j_t)}
\end{eqnarray}
with each $x^{j_l}_{(j_1, j_2, \cdots, j_t)}$ indicating the number of points assigned to the $j_l$-th cluster from $\mathcal{R}_{(j_1, j_2, \cdots, j_t)}$. Thus, we have the following two types of linear constraints: for each region $\mathcal{R}_{(j_1, j_2, \cdots, j_t)}$,
\begin{eqnarray}
x^{j_1}_{(j_1, j_2, \cdots, j_t)}+ \cdots +x^{j_t}_{(j_1, j_2, \cdots, j_t)}=n_{(j_1, j_2, \cdots, j_t)}, \label{for-llp1}
\end{eqnarray}
and for each $j_l\in\{1, 2, \cdots, k\}$,
\begin{eqnarray}
L\leq \sum_{(j_1, j_2, \cdots, j_t)\in \pi_{j_l}} x^{j_l}_{(j_1, j_2, \cdots, j_t)}\leq U. \label{for-llp2}
\end{eqnarray}
Here $\pi_{j_l}$ is the set of all the possible subsets containing $j_l$ of $\{1, \cdots, k\}$.

%
%
%

\begin{figure}[h]
\centering
 \includegraphics[height=2in]{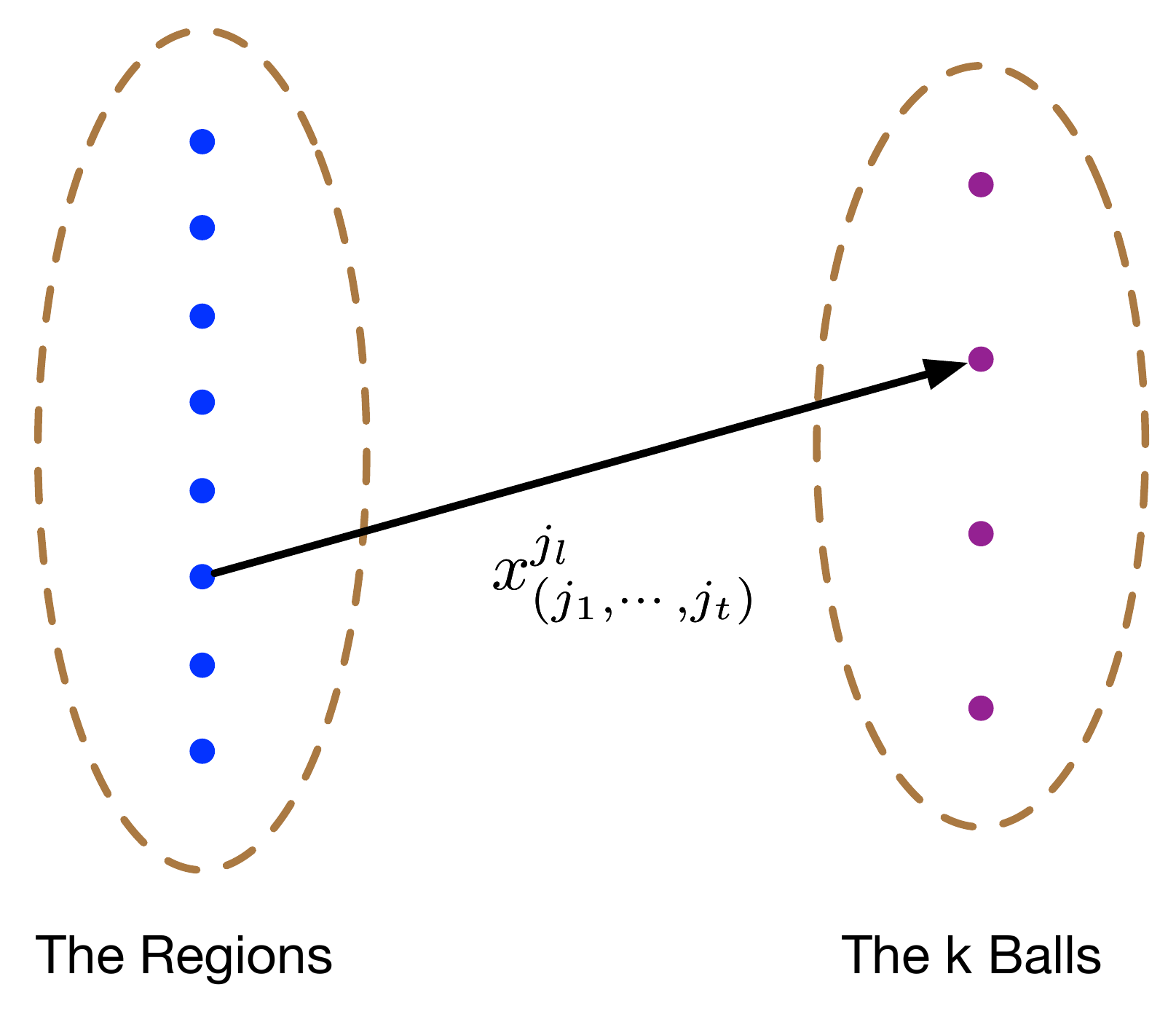}
  \caption{Each region has the supply amount $n_{(j_1, j_2, \cdots, j_t)}$, and each ball has the capacity $U$ and demand $L$.} 
   \label{fig-kbc_match}
\end{figure}
First, it is straightforward to have the following lemma.

\begin{lemma}
\label{lem-soltopartition}
An integral solution of (\ref{for-llp1})-(\ref{for-llp2}) corresponds to a balanced partition of $P$ covered by the $k$ balls $\mathcal{B}_1, \cdots, \mathcal{B}_k$, and vice versa.
\end{lemma}

In fact, solving the SoL (\ref{for-llp1})-(\ref{for-llp2}) is exactly a problem of max flow (but the complexity is independent of $n$). We build a bipartite graph $G=(V, E)$ with $V=V_{\mathcal{R}}\cup V_{\mathcal{B}}$: $V_{\mathcal{R}}$ includes a set of vertices with each corresponding to an individual region $\mathcal{R}_{(j_1, j_2, \cdots, j_t)}$, and $V_{\mathcal{B}}$ includes a set of $k$ vertices corresponding to the $k$ balls (with a slight abuse of notation, we still use $\mathcal{R}_{(j_1, j_2, \cdots, j_t)}$s and $\mathcal{B}_j$s to denote the vertices); for each variable $x^{j_l}_{(j_1, j_2, \cdots, j_t)}$, we connect an edge from $\mathcal{R}_{(j_1, j_2, \cdots, j_t)}$ to $\mathcal{B}_{j_l}$ with the value of $x^{j_l}_{(j_1, j_2, \cdots, j_t)}$ being the flow on the edge; each $\mathcal{R}_{(j_1, j_2, \cdots, j_t)}$ has the supply amount $n_{(j_1, j_2, \cdots, j_t)}$, and each $\mathcal{B}_j$ has the capacity $U$ and demand $L$. See Figure.~\ref{fig-kbc_match}. The SoL (\ref{for-llp1})-(\ref{for-llp2}) is feasible iff the max flow has the amount $n$. With both capacities and demands, the problem can be transformed to a typical max flow problem with only capacities, which can be solved in $O(|V|\cdot |E|)$ time~\cite{erickson}. Note that the SoL (\ref{for-llp1})-(\ref{for-llp2}) has at most $k2^{k}$ variables and $|V|<2^k+k$. Thus, the time complexity 
\begin{eqnarray}
O(|V|\cdot |E|)=O(k 2^{2k}). \label{for-matchtime}
\end{eqnarray}

Once obtaining a max flow of the bipartite graph, we further check that whether it is an integral solution. If it has some fractional flow values, we need to transform it to an integral solution without decreasing the total flow. Let us remove all the edges having integral values and only focus on the remaining ones in the bipartite graph. We conduct the following strategy to eliminate the edges until empty, that is, an integral solution is obtained. Starting from an arbitrarily picked edge, we grow it to be a path from the edge's two sides, until one of the two cases happens: \textbf{(1)} the path contains a circle or \textbf{(2)} no more edge can be added (i.e., the two endpoints are both degree-$1$ vertices). See Figure.~\ref{fig-case1} and \ref{fig-case2}. Since there are only $k$ vertices (i.e., the $k$ balls) in the right column of the bipartite graph (Figure.~\ref{fig-kbc_match}), the resulting path has at most $O(k)$ vertices.

\begin{figure}[h]
\vspace{-.2cm} 
\centering
  \subfloat[]{\label{fig-case1}\includegraphics[height=1.7in]{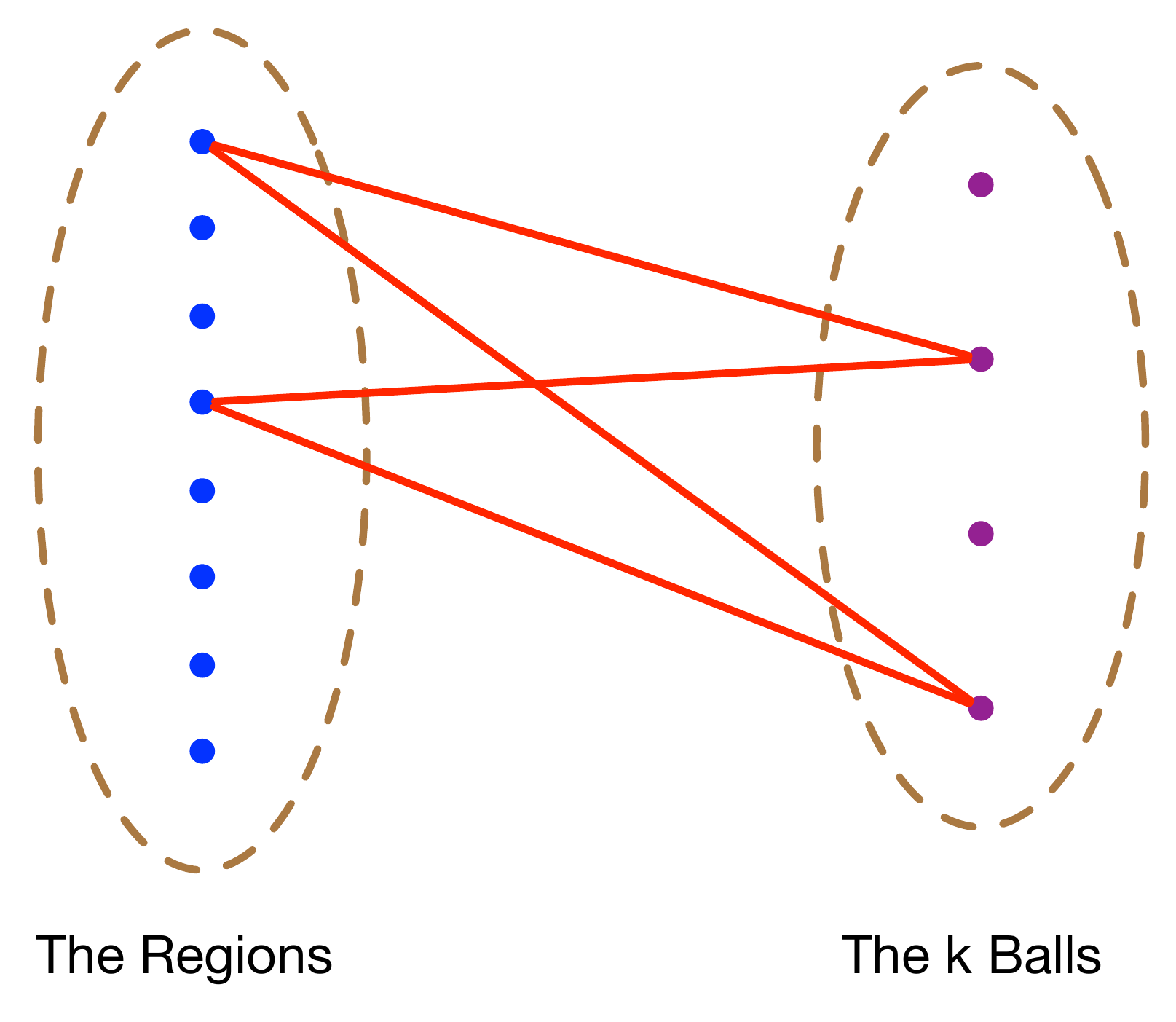}}
  \hspace{0.5in}
  \subfloat[]{\label{fig-case2}\includegraphics[height=1.7in]{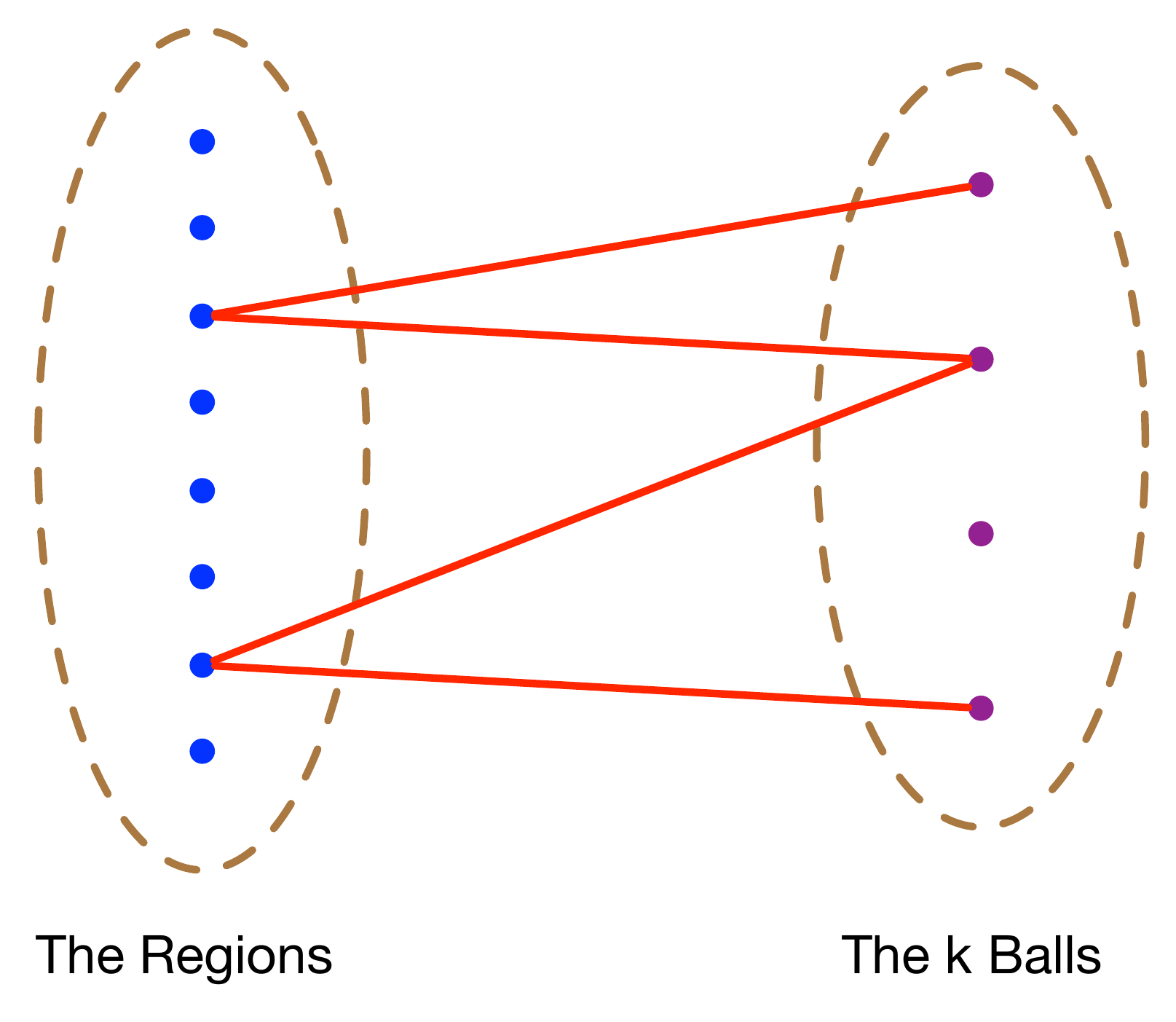}}
  \caption{The two cases when growing the path.} 
  \vspace{-0.25in}
\end{figure}

Without loss of generality, we denote the circle in case (1) as 
\begin{eqnarray}
\big\{v_{\mathcal{R}_1}, \hspace{0.06in}v_{\mathcal{B}_1}, \hspace{0.06in} v_{\mathcal{R}_2}, \hspace{0.06in}\cdots, \hspace{0.06in} v_{\mathcal{R}_h}, \hspace{0.06in} v_{\mathcal{B}_h}, \hspace{0.06in} v_{\mathcal{R}_1}\big\}, \label{for-newcircle}
\end{eqnarray}
where each $v_{\mathcal{R}_j}$ and $v_{\mathcal{B}_j}$ denote the vertices in the left and right column, respectively. We also assume the corresponding fractional flow values (i.e., the variables) are
\begin{eqnarray}
\big\{x^1_{*_1}, \hspace{0.06in} x^1_{*_2}, \hspace{0.06in} x^2_{*_2}, \hspace{0.06in} x^2_{*_3}, \hspace{0.06in} \cdots, \hspace{0.06in} x^h_{*_h}, \hspace{0.06in} x^h_{*_1}\big\}.\label{for-variable-case1}
\end{eqnarray}
Each $x^i_{*_j}$ indicates the flow value from $\mathcal{R}_j$ to $\mathcal{B}_i$. Here we replace the foot subscript by $*_j$ to simplify our analysis. 
Meanwhile, we choose the positive value 
\begin{eqnarray}
\delta=\min\{x^j_{*_j}-\lfloor x^j_{*_j}\rfloor, \lceil x^{j-1}_{*_{j}}\rceil -x^{j-1}_{*_{j}}\mid 1\leq j\leq h\} \label{for-round}
\end{eqnarray}
where $x^0_{*_1}$ represents $x^h_{*_1}$ for convenience. Then, the following numbers
\begin{eqnarray}
x^1_{*_1}-\delta,x^1_{*_2}+\delta, x^2_{*_2}-\delta,x^2_{*_3}+\delta, \cdots, x^h_{*_h}-\delta,x^h_{*_1}+\delta
\end{eqnarray}
contain at least one integer and all the others remain non-negative. More importantly, the total flow passing through each vertex remains the same; that is, the modified solution is still a max flow in the bipartite graph.  

For case (2), we first claim that both the two endpoints must be in the right column of $G$ (see Figure.~\ref{fig-case2}). Since the current solution is already a max flow with the amount $n$, the total flow passing through each vertex $v_{\mathcal{R}_j}$ in the left column should exactly match the supply amount $n_{*_j}$ which is an integer. In addition, we know each endpoint has only one fractional-flow-value edge. Therefore, if an endpoint is in the left column, it will be contradict to the fact that its total flow $n_{*_j}$ is an integer. Then, we suppose that the case (2) path is 
\begin{eqnarray}
\big\{v_{\mathcal{B}_1}, \hspace{0.06in}v_{\mathcal{R}_1}, \hspace{0.06in} v_{\mathcal{B}_2}, \hspace{0.06in}\cdots, \hspace{0.06in} v_{\mathcal{R}_{h-1}}, \hspace{0.06in} v_{\mathcal{B}_h},\big\},
\end{eqnarray}
and the corresponding fractional flow values (i.e., the variables) are
\begin{eqnarray}
\big\{x^1_{*_1}, \hspace{0.06in} x^2_{*_1}, \hspace{0.06in} x^2_{*_2}, \hspace{0.06in} x^3_{*_2}, \hspace{0.06in} \cdots, \hspace{0.06in} x^h_{*_{h-1}}\big\}.
\end{eqnarray}
Let the total flows passing through $v_{\mathcal{B}_1}$ and $v_{\mathcal{B}_h}$ be $m_1$ and $m_h$ respectively. Since $v_{\mathcal{B}_1}$ and $v_{\mathcal{B}_h}$ both have only one fractional-flow-value edge, we know that $m_1$ and $m_h$ are fractional values. We choose the positive value 
\begin{eqnarray}
\delta=\min\Big\{\{m_1-L, U-m_h\}\cup\{x^j_{*_j}-\lfloor x^j_{*_j}\rfloor, \lceil x^{j+1}_{*_{j}}\rceil -x^{j+1}_{*_{j}}\mid 1\leq j\leq h-1\} \Big\} .
\end{eqnarray}
Similar to case (1),  the following numbers
\begin{eqnarray}
x^1_{*_1}-\delta,x^2_{*_1}+\delta, x^2_{*_2}-\delta,x^3_{*_2}+\delta, \cdots, x^{h-1}_{*_{h-1}}-\delta,x^h_{*_{h-1}}+\delta
\end{eqnarray}
contain at least one integer and no constraint of the SoL (\ref{for-llp1})-(\ref{for-llp2}) is violated after this adjustment.

Overall, the adjustment for case (1) or (2) adds at least one new integral flow, and we can remove its corresponding edge from $G$. After at most $|E|=k2^{k}$ times, $G$ will have no edge. In other words, an integral solution of the SoL (\ref{for-llp1})-(\ref{for-llp2}) is obtained. Furthermore, each adjustment costs $O(k)$ time and the  complexity of all the adjustments is 
\begin{eqnarray}
|E|\cdot O(k)=O(k^2 2^k).\label{for-roundtime}
\end{eqnarray}
Combining our above analysis and the time complexities (\ref{for-matchtime}) and (\ref{for-roundtime}), we have the following theorem.


\begin{theorem}
\label{lem-integer}
We can compute an integral solution of (\ref{for-llp1})-(\ref{for-llp2}) in $O(k^2 2^{2k})$ time. 
\end{theorem}

\subsection{Balanced Partition for $k$-BMedian and $k$-BMeans}
\label{sec-ext-select}
%
%
%
In this section, we extend the idea of Section~\ref{sec-feasible} to compute the best balanced partition for $k$-BMedian and $k$-BMeans. We also assume that the 
$k$ cluster centers $\{c_1, \cdots, c_k\}$ are fixed. 
Different to the SoL (\ref{for-llp1})-(\ref{for-llp2}) for $k$-BCenter, we need to add an objective function to minimize the sum of distances (resp., squared distances) for the $k$-BMedian (resp., $k$-BMeans). We build a linear programming model instead. 

To better illustrate our idea, we refine the spatial partition in Section~\ref{sec-feasible} with respect to a given small $\epsilon>0$. Suppose we compute the $kn$ pairwise distances $\{||p-c_j||\mid p\in P, 1\leq j\leq k\}$ in advance, and let the largest and smallest distances (except the $0$ distance, if it exists) be $r_{max}$ and $r_{min}$ respectively. Without loss of generality, we assume $L=\log_{1+\epsilon}\frac{r_{max}}{r_{min}}$ is an integer. For each $c_j$, we draw 
\begin{eqnarray}
L+1=O( \frac{1}{\epsilon}\log \frac{r_{max}}{r_{min}})
\end{eqnarray}
balls $\mathcal{B}_{j, 0}, \cdots, \mathcal{B}_{j, L}$ co-centered at $c_j$, where the radii are 
\begin{eqnarray}
r_{min},\hspace{0.1in} (1+\epsilon)r_{min}, \hspace{0.1in} (1+\epsilon)^2r_{min},\cdots, \hspace{0.1in} (1+\epsilon)^{L}r_{min}=r_{max} 
\end{eqnarray}
correspondingly. 
Overall, the $k(L+1)$ balls $\{\mathcal{B}_{j,l}\mid 1\leq j\leq k, 0\leq l\leq L\}$ partition the Euclidean space into at most $(L+1)^k=O(L^k)$ regions: let $0\leq l_j\leq L$ for each $1\leq j\leq k$, the corresponding region is denoted as 
\begin{eqnarray}
\mathcal{R}_{(l_1, \cdots,  l_k)}= (\mathcal{B}_{1, l_1}\setminus\mathcal{B}_{1, l_1-1})\cap\cdots\cap(\mathcal{B}_{k, l_k}\setminus\mathcal{B}_{k, l_k-1}). \label{for-ext-region}
\end{eqnarray}
If $l_j=0$, we simply let $\mathcal{B}_{j, l_j-1}$ be the single point $c_j$ in (\ref{for-ext-region}). We ignore the region outside $\cap^k_{j=1}\mathcal{B}_{j, L}$, since no point locates there (due to the definition of $r_{max}$). Note that the definition of $\mathcal{R}_{(l_1, \cdots,  l_k)}$ is different from the that of $\mathcal{R}_{(j_1, j_2, \cdots, j_t)}$ in Section~\ref{sec-feasible}.

\begin{figure}[h]
\centering
 \includegraphics[height=1.7in]{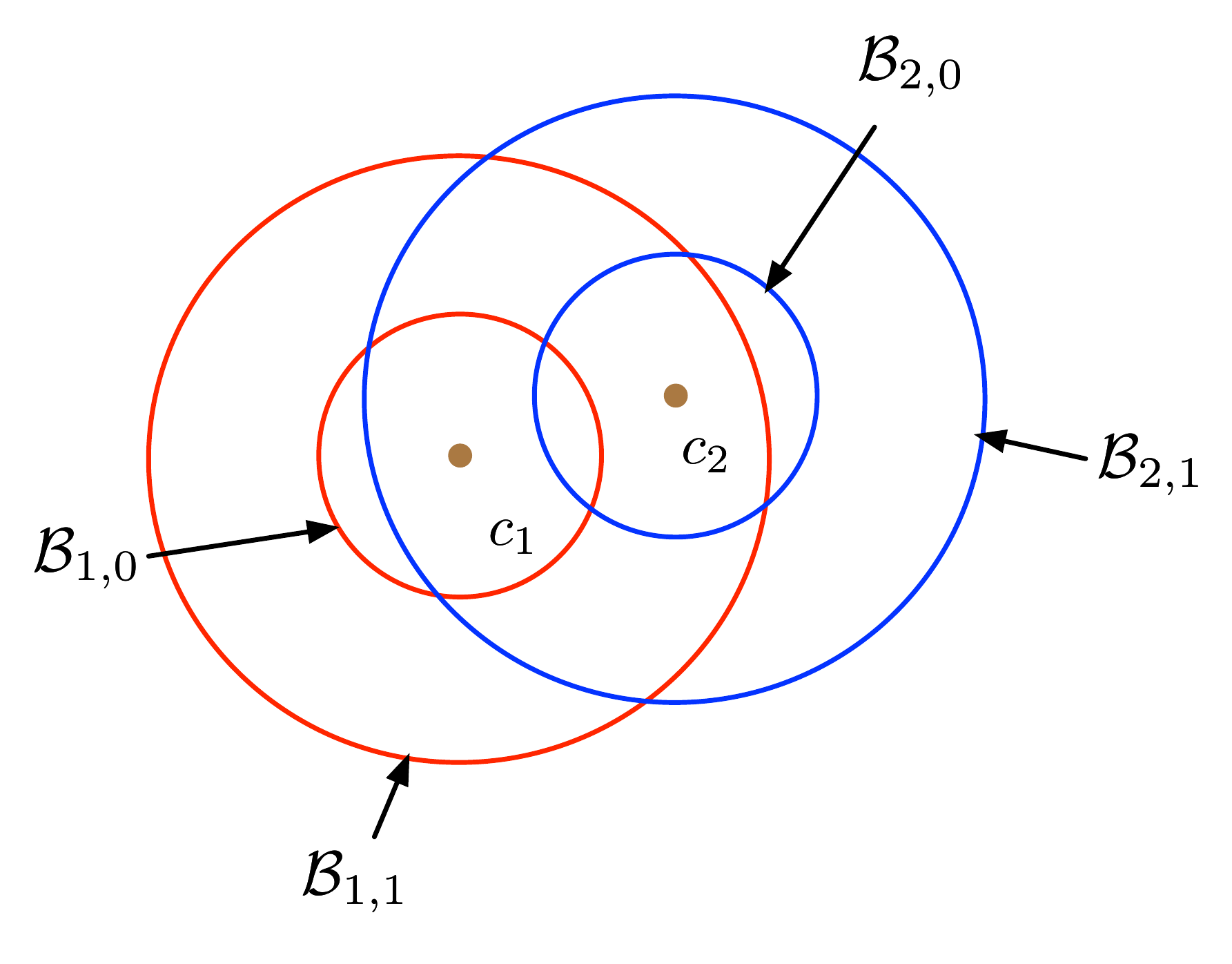}
  \caption{We show an example with $k=2$ and $L=1$. Some regions may not be connected, e.g., $\mathcal{R}_{(1,1)}$.} 
   \label{fig-partition}
\end{figure}

See Figure~\ref{fig-partition} for an illustration. For each $\mathcal{R}_{(l_1, \cdots,  l_k)}$, denote by $n_{(l_1, \cdots,  l_k)}$ the number of covered points.
 The intuition behind our spatial partition is to discretize those $kn$ distances, such that we can handle the problem more conveniently. 

Similar to the SoL in Section~\ref{sec-feasible}, we assign $k$ non-negative variables 
\begin{eqnarray}
x^{1}_{(l_1, l_2, \cdots, l_k)},\hspace{0.1in} \cdots,\hspace{0.1in} x^{k}_{(l_1, l_2, \cdots, l_k)}
\end{eqnarray}
for each region  $\mathcal{R}_{(l_1, \cdots,  l_k)}$; each $x^{j}_{(l_1, l_2, \cdots, l_k)}$ indicates the number of points assigned to the $j$-th cluster from $\mathcal{R}_{(l_1, \cdots,  l_k)}$. The index $l_j$ is called the \textbf{``level order''} of the variable $x^{j}_{(l_1, l_2, \cdots, l_k)}$, which indicates that how far from those points to their cluster center. In addition, we have the coefficients $\{\alpha_0, \alpha_1, \cdots, \alpha_L\}$ with each $\alpha_{l}=(1+\epsilon)^{l}r_{min}$ for $0\leq l\leq L$. Denote by $\pi$ the set of $k$-tuples $\{(l_1, \cdots, l_k)\mid 0\leq l_j\leq L, \forall 1\leq j\leq k\}$. Then, we have the following linear programming (LP).
\begin{eqnarray}
&&\min\sum_{(l_1, \cdots, l_k)\in\pi}\sum^k_{j=1}\alpha_{l_j}x^j_{(l_1, l_2, \cdots, l_k)},\label{for-lp1}\\
\forall (l_1, l_2, \cdots, l_k)\in\pi, &&\hspace{0.1in} x^1_{(l_1, l_2, \cdots, l_k)}+\cdots+x^k_{(l_1, l_2, \cdots, l_k)}=n_{(l_1, \cdots,  l_k)},\label{for-lp2}\\
\forall j=1, \cdots, k, &&\hspace{0.1in} L\leq \sum_{(l_1, l_2, \cdots, l_k)\in\pi} x^j_{(l_1, l_2, \cdots, l_k)}\leq U. \label{for-lp3}
\end{eqnarray}

It is easy to have the complexity of the LP.

\begin{lemma}
\label{lem-lpsize}
The  LP (\ref{for-lp1})-(\ref{for-lp3}) has $k(L+1)^k$ variables and $O(L^k)$ constraints in total.
\end{lemma}

More importantly, we reveal the relation between the LP model and $k$-BMedian in Lemma~\ref{lem-lp2} (we will show that the result can be easily extended for $k$-BMeans later).

\begin{lemma}
\label{lem-lp2}
\textbf{(\rmnum{1})} Any balanced partition on $P$ corresponds to a feasible integral solution of (\ref{for-lp2})-(\ref{for-lp3}), and vice versa. \textbf{(\rmnum{2})} If the resulting $k$-BMedian cost of a balanced partition is $W\geq 0$, the corresponding feasible integral solution will have the objective value of (\ref{for-lp1}) within $[W, (1+\epsilon)W)$.
\end{lemma}
\begin{proof}
We prove \textbf{(\rmnum{1})} first. Given a balanced partition on $P$, we set each variable $x^j_{(l_1, l_2, \cdots, l_k)}$ to be the number of points assigned to the $j$-th cluster from the region $\mathcal{R}_{(l_1, \cdots,  l_k)}$. Because the partition is balanced, the constraints (\ref{for-lp2}) and (\ref{for-lp3}) are satisfied, i.e, a feasible integral solution is obtained. Given a feasible integral solution of (\ref{for-lp2})-(\ref{for-lp3}), we can simply generate a balanced partition as follows. For each variable $x^j_{(l_1, l_2, \cdots, l_k)}$, we just arbitrarily pick $x^j_{(l_1, l_2, \cdots, l_k)}$ points from the region $\mathcal{R}_{(l_1, \cdots,  l_k)}$ and assign them to the $j$-th cluster. The constraints (\ref{for-lp2}) and (\ref{for-lp3}) guarantee that the obtained partition is balanced. So \textbf{(\rmnum{1})} is true.

Suppose the given balanced partition has the clustering cost $W$, and we consider it corresponding feasible integral solution. From \textbf{(\rmnum{1})}, we know that there is a set of $x^j_{(l_1, l_2, \cdots, l_k)}$ points assigned to the $j$-th cluster. Moreover, the contribution of these $x^j_{(l_1, l_2, \cdots, l_k)}$ points to the objective function (\ref{for-lp1}) is $\alpha_{l_j}x^j_{(l_1, l_2, \cdots, l_k)}$. In fact, we know their contribution to the clustering cost should be within $(\frac{1}{1+\epsilon}\alpha_{l_j}x^j_{(l_1, l_2, \cdots, l_k)},\alpha_{l_j}x^j_{(l_1, l_2, \cdots, l_k)}]$. Overall, the objective value of (\ref{for-lp1}) should be within $[W, (1+\epsilon)W)$.
\qed
\end{proof}

Lemma~\ref{lem-lp2} directly implies the following result. 

\begin{lemma}
\label{lem-lp1}
The optimal objective value in (\ref{for-lp1}) is at most $(1+\epsilon)$ times the minimum $k$-BMedian cost induced by the given $k$ cluster centers $\{c_1, \cdots, c_k\}$.
\end{lemma}
%
%
%
%
%
%
%
%
%
%
%
%
%
%
\begin{remark}
Lemma~\ref{lem-lp1} can be extended for $k$-BMeans by two modifications. First, we need to replace each coefficient $\alpha_l$ by $\alpha^2_l$ in the objective function (\ref{for-lp1}). In addition, the approximation factor ``$1+\epsilon$'' becomes ``$(1+\epsilon)^2=1+O(\epsilon)$'' in Lemma~\ref{lem-lp1}.
\end{remark}

Similar to Section~\ref{sec-feasible}, the remaining question is to solve (\ref{for-lp1})-(\ref{for-lp3}) and guarantee that the optimal solution is an integral solution. We build a similar graph as Figure.~\ref{fig-kbc_match}, where the number of vertices in the left column becomes larger since we have $O(L^k)$ regions now; we still have $k$ vertices in the right column, and each vertex corresponds to an individual cluster center $c_j$. Thus, $|E|=O(k L^k)$ and $|V|=O(L^k)$. In addition, the object here is finding a minimum cost maximum flow instead: each edge connecting  $\mathcal{R}_{(l_1, \cdots,  l_k)}$ and $c_j$ has the cost $\alpha_{l_j}$ (or $\alpha^2_l$ for $k$-BMeans). An optimal solution can be achieved in 
\begin{eqnarray}
O(|V|\cdot |E|^2\cdot \log^2 |V|)=O(k^4 L^{3k}\log^2 L) \label{for-matchtime2}
\end{eqnarray}
time by existing algorithms~\cite{erickson}.

To transform the solution to an integral solution, we adopt the same strategy in Section~\ref{sec-feasible}. The only place we need to pay more attention is that our adjustments should not cause any loss in the objective value (\ref{for-lp1}).  
 We take the case (1) as an example and use the same notations in~(\ref{for-newcircle}) and (\ref{for-variable-case1}). For ease of understanding, we denote the level orders of $x^j_{*_j}$ and $x^{j-1}_{*_{j}}$ as $l_j$ and $l'_{j-1}$ respectively ($l'_0$ indicates $l'_h$). 

\begin{lemma}
\label{lem-ext-round}
If case (1) happens, we have $\sum^h_{j=1}(\alpha_{l'_{j-1}}-\alpha_{l_j})=0$.
\end{lemma}
\begin{proof}
Different to (\ref{for-round}), we can consider two  directions to find the small number $\delta$. Let
\begin{eqnarray}
\delta_1&=&\min\{x^j_{*_j}-\lfloor x^j_{*_j}\rfloor, \lceil x^{j-1}_{*_{j}}\rceil -x^{j-1}_{*_{j}}\mid 1\leq j\leq h\};\\
\delta_2&=&\min\{\lceil x^j_{*_j}\rceil- x^j_{*_j}, x^{j-1}_{*_{j}} -\lfloor x^{j-1}_{*_{j}}\rfloor\mid 1\leq j\leq h\}.
\end{eqnarray}

If $\sum^h_{j=1}(\alpha_{l'_{j-1}}-\alpha_{l_j})<0$, we can safely round the variables 
\begin{eqnarray}
x^1_{*_1}-\delta_1,x^1_{*_2}+\delta_1, x^2_{*_2}-\delta_1,x^2_{*_3}+\delta_1, \cdots, x^h_{*_h}-\delta_1,x^h_{*_1}+\delta_1
\end{eqnarray}
without violating any constraint. Moreover, the objective value of (\ref{for-lp1}) is decreased by $\big(\sum^h_{j=1}(\alpha_{l_j}-\alpha_{l'_{j-1}})\big)\delta_1>0$. 

For the other case, if $\sum^h_{j=1}(\alpha_{l'_{j-1}}-\alpha_{l_j})>0$, we can round the variables to the other direction to be 
\begin{eqnarray}
x^1_{*_1}+\delta_2,x^1_{*_2}-\delta_2, x^2_{*_2}+\delta_2,x^2_{*_3}-\delta_2, \cdots, x^h_{*_h}+\delta_2,x^h_{*_1}-\delta_2,
\end{eqnarray}
and the objective value is decreased by $\big(\sum^h_{j=1}(\alpha_{l'_{j-1}}-\alpha_{l_j})\big)\delta_2>0$.

These two cases are both contradict to the fact that we already obtain the optimal solution of (\ref{for-lp1})-(\ref{for-lp3}). Thus,  $\sum^h_{j=1}(\alpha_{l'_{j-1}}-\alpha_{l_j})=0$.
\qed
\end{proof}

According to Lemma~\ref{lem-ext-round}, we can freely round the variables along a case (1) circle to either direction, and add at least one more integer variable. For case (2), we can prove the same result by the same manner. Similar to (\ref{for-roundtime}), all the adjustments cost 
\begin{eqnarray}
|E|\cdot O(k)=O(k^2 L^k).\label{for-roundtime2}
\end{eqnarray}
Combining our above analysis and the time complexities (\ref{for-matchtime2}) and (\ref{for-roundtime2}), we have the following theorem.
%
\begin{theorem}
\label{lem-ext-integer}
We can compute an integral optimal solution for the LP (\ref{for-lp1})-(\ref{for-lp3}) in $O(k^4$ $ L^{3k}$ $\log^2 L)$ time. 
\end{theorem}

\section{Balanced Clustering Algorithms}
\label{sec-high}

In this section, we propose our algorithms for $k$-BCenter, $k$-BMedian, and $k$-BMeans separately. In general, we need to find the qualified candidates for cluster centers, and then apply the ideas in Section~\ref{sec-partition}  to form the balanced clusterings.


\subsection{Algorithm for $k$-BCenter}
\label{sec-candidate}

Gonazlez's seminal paper~\cite{G85} provided an elegant $2$-approximation algorithm for $k$-center clustering in any dimension. Basically, the algorithm iteratively selects $k$ points from the input, where the initial point is arbitrarily selected, and each following $j$-th step ($2\leq j\leq k$) chooses the point having the largest minimum distance to the already selected $j-1$ points. Finally, it is able to show that these $k$ points induce a $2$-approximation for $k$-center clustering if each input point is assigned to its nearest neighbor of these $k$ points. 

Denote by $S=\{s_1, s_2, \cdots, s_k\}$ these ordered $k$ points selected by Gonazlez's algorithm, and define the Cartesian product $\underbrace{S\times\cdots\times S}_k$ as $S^k$, i.e., 
\begin{eqnarray}
S^k=\{(s'_1, s'_2, \cdots, s'_k)\mid s'_j\in S, 1\leq j\leq k\}.
\end{eqnarray} 
Then we have the following lemma. 

\begin{lemma}
\label{lem-4app}
There exists a $k$-tuple points from $S^k$ yielding a $4$-approximation for $k$-BCenter. 
\end{lemma}
\begin{proof}
Suppose the unknown $k$ optimal balanced clusters are $C_1, C_2, \cdots, C_k$, and the optimal radius is $r_{opt}$. If the selected $k$ points of $S$ luckily fall to these $k$ clusters separately, it is easy to obtain a $2$-approximation through triangle inequality, and the balanced clusters can be formed by the partition idea in Section~\ref{sec-feasible}.

Now, we consider the other case.  Suppose that $(s_{j_1},s_{j_2})$ is the firstly appeared pair belonging to the same optimal cluster and $j_1<j_2$. Without loss of generality, we assume that $s_j\in C_j$ for $1\leq j\leq j_2-1$. Due to the nature of Gonazlez's algorithm, we know that 
\begin{eqnarray}
\max_{p\in \cup^k_{j=j_2}C_j}\big\{\min_{1\leq l\leq j_2-1}||p-s_l||\big\}\leq ||s_{j_1}-s_{j_2}||\leq 2r_{opt}.\label{for-4app1}
\end{eqnarray}
Note that (\ref{for-4app1}) cannot directly imply a $2$-approximation of $k$-BCenter. For instance, we cannot simply assign each point to its nearest neighbor of $\{s_1,  \cdots,$ $ s_{j_2-1}\}$, because of the requirement of balance;
%
actually, this is also the major difference between the ordinary and balanced clustering problems. 
Instead, for each $j\geq j_2$, we arbitrarily select a point $p\in C_j$ and let its nearest neighbor of $\{s_1, \cdots,$ $ s_{j_2-1}\}$ be $s_{l(j)}$; then we assign it as the cluster center of  $C_j$. Correspondingly, for any $q\in C_j$ we have
\begin{eqnarray}
||q-s_{l(j)}||\leq ||q-p||+||p-s_{l(j)}||\leq 4r_{opt} \label{for-4app2}
\end{eqnarray}
due to triangle inequality and the fact that both $||q-p||$ and $||p-s_{l(j)}||$ are no larger than $2r_{opt}$. Thus, the $k$-tuple points $\{s_1,$ $ s_2,$ $ \cdots,$ $ s_{j_2-1},$ $ s_{l(j_2)},$ $ \cdots, s_{l(k)}\}$ yields a $4$-approximation if each optimal cluster $C_j$ takes the $j$-th point in the tuple as its cluster center. 
\qed
\end{proof}
\begin{remark}
In Appendix, we construct the examples to show that (1) the approximation factor $4$ is tight enough and (2) it is necessary to use $S^k$ rather than the simple $S$.
\end{remark}



Combining Lemma~\ref{lem-4app} and the partition idea in Section~\ref{sec-feasible}, we have Algorithm~\ref{alg-1} for $k$-BCenter. Step 1 and 2 take $O\big(knd+nk\log (nk)\big)$ time, and step 3 runs at most $O(k^k\log n)$ rounds with each round costing $O\big(n+k^2 2^{2k}\big)$ time. Thus, the total running time is $O\big(n(\log n +d)\big)$ if $k$ is a constant.

\begin{algorithm}[tb]
   \caption{$k$-BCenter}
   \label{alg-1}
\begin{algorithmic}
   \STATE {\bfseries Input:} $P=\{p_i, \mid 1\leq i\leq n\}\subset\mathbb{R}^d$, an integer $k\geq 1$, and the integer lower and upper bounds $1\leq L\leq U\leq n$.
\begin{enumerate}
\item Run Gonazlez's algorithm and output $k$ points $S=\{s_1, s_2, \cdots, s_k\}$. 
\item Compute the $nk$ distances from $P$ to $S$, and sort them in an increasing order. Denote by $\mathcal{R}$ the set of distances.
\item For each $k$-tuple $(s'_1, \cdots, s'_k)$ from $S^k$, binary search on $\mathcal{R}$. Initialize the optimal radius $r_{opt}=\max\mathcal{R}$. For each step with $r\in \mathcal{R}$, do the following steps.
\begin{enumerate}
\item Draw the $k$ balls with radii $r$ and centered at $(s'_1, \cdots, s'_k)$ separately. 
\item If the SoL (\ref{for-llp1})-(\ref{for-llp2}) is feasible, 
\begin{itemize}
\item if $r<r_{opt}$, update $r_{opt}$ to be $r$  and record the feasible solution; 
\item if $r$ is not a leaf, continue the binary search to the left side. Else, stop binary search.
\end{itemize}
\item Else, 
\begin{itemize}
\item if $r$ is not a leaf, continue the binary search to the right side. Else, stop binary search.
\end{itemize}
\end{enumerate}
\item Return the $k$-tuple from $S^k$ with the smallest $r_{opt}$ associating the corresponding feasible solution, and transform it to be an integral solution via Theorem~\ref{lem-integer}.
\end{enumerate}
\end{algorithmic}
\end{algorithm}

\begin{theorem}
\label{the-4app}
Algorithm~\ref{alg-1} yields a $4$-approximation of $k$-BCenter, and the running time is $O\big(n(\log n+d)\big)$ if $k$ is a constant. 
\end{theorem}

Since the $k$ cluster centers are always chosen from the input $P$ in our algorithm, the result can be directly extended to Metric $k$-BCenter.

\begin{corollary}
\label{cor-4app}
Given an instance of Metric $k$-BCenter, we suppose the time complexity for acquiring the distance between any two vertices is $O(D)$. Algorithm~\ref{alg-1} yields a $4$-approximation and the running time is $O\big(n(\log n+D)\big)$ if $k$ is a constant. 
\end{corollary}
\begin{remark}
Usually, a $n$-point metric is of size $\Theta(n^2)$ if all the ${n\choose 2}$ pairwise distances are given (i.e., $O(D)=O(1)$). Thus, our algorithm has sub-linear complexity in terms of the input size.
\end{remark}


\subsection{ $O(1)$-Approximation Algorithms for $k$-BMedian and $k$-BMeans}
\label{sec-ext-candidate}
%
%
First, we need to study that how to obtain the qualified $k$ cluster centers for $k$-BMedian and $k$-BMeans.
Similar to Lemma~\ref{lem-4app}, we also use the Cartesian product to bridge the ordinary and balanced clusterings.
%
%
%
%
%
%
The following lemma was originally discovered by Ding and Xu~\cite{DBLP:conf/soda/DingX15}, and we slightly modify the statement to make it more suitable for our problem.

\begin{lemma}[\cite{DBLP:conf/soda/DingX15}]
\label{lem-bridge1}
Let  $P$ be an instance of  $k$-BMedian (resp., $k$-BMeans). Suppose $\lambda\geq 1$, and $S=\{s_{1},\cdots,s_{k}\}$ is the set of the cluster centers yielding a $\lambda$-approximation of the ordinary $k$-median (resp., $k$-means) clustering on $P$ (without considering the balance). Then the Cartesian product $S^k$ contains at least one $k$-tuple that induces a $(3\lambda+2)$ (resp., $(18\lambda+16)$)-approximation of $k$-BMedian (resp., $k$-BMeans).
\end{lemma}

Actually, a number of approximation algorithms for the ordinary $k$-median and $k$-means clusterings have been studied before. In particular, Indyk~\cite{indyk1999sublinear} provided a nearly linear time bi-criteria approximation algorithm for metric $k$-median clustering, where it outputs $O(k)$ cluster centers and yields a clustering cost within a constant factor of the optimal cost. Further, Chen~\cite{chen2009coresets} improved the running time to be linear $O(knd)$ and show that the algorithm can handle both $k$-median and $k$-means clusterings in Euclidean or abstract metric space. Their algorithms use a slow bi-criteria approximation algorithm (e.g., \cite{jain2001approximation,charikar2005improved}) as the black-box to deal with small random samples. Thus the exact approximation factor depends on which black-box algorithm they use; the factor also depends on the type of the given clustering problem, such as (Metric) $k$-BMedian or (Metric) $k$-BMeans. For simplicity, we just denote the approximation factor as the unified $O(1)$. We refer the reader to~\cite{indyk1999sublinear,chen2009coresets} for more details. Further, the following Lemma~\ref{lem-bridge2} shows that we can obtain the $O(1)$-approximations for $k$-BMedian and $k$-BMeans through the bi-criteria approximations.


\begin{lemma}
\label{lem-bridge2}
Let $P$ be an instance of  $k$-BMedian (resp., $k$-BMeans). We first ignore the requirement of balance and run the algorithm~\cite{chen2009coresets} on $P$ to obtain the set $C$ of $O(k)$ cluster centers. Then, the Cartesian product $C^k$ contains at least one $k$-tuple yielding an $O(1)$-approximation for $k$-BMedian (resp., $k$-BMeans) on $P$.
\end{lemma}
\begin{proof}
We focus on $k$-median clustering first. In fact, the set $C$ can be viewed as a new instance of the $k$-median clustering, where each point of $C$ is a multi-set of points with the size equal to the size of the corresponding cluster in the bi-criteria approximation. We restrict the cluster centers to be selected from $C$, and let $\tilde{C}\subset C$ be the best selection, i.e., the $k$-tuple yielding the smallest clustering cost among all the ${|C|\choose k}$ choices. Note that the optimal cluster centers in fact do not necessarily come from $C$ in Euclidean space. So $\tilde{C}$ yields a $2$-approximation on $C$ by triangle inequality. Also, we know that the set $C$ yields a clustering cost within a constant factor of the optimal cost on $P$. Using triangle inequality again, we know that $\tilde{C}$ is an $O(1)$-approximation of the $k$-median clustering on $P$. The similar result can be easily extended to $k$-means clustering, where the only difference is that the triangle inequality is replaced by the weaker one, i.e., $||a+b||^2\leq 2||a||^2+2||b||^2$ for any vectors $a$ and $b$.

Through Lemma~\ref{lem-bridge1}, we know that $\tilde{C}^k$ contains at least one $k$-tuple yielding an $O(1)$-approximation for $k$-BMedian (resp., $k$-BMeans) on $P$. Since $\tilde{C}^k\subset C^k$, Lemma~\ref{lem-bridge2} is true.
\qed
\end{proof}


\begin{algorithm}[tb]
   \caption{$O(1)$-Approximation Algorithm for  $k$-BMedian/$k$-BMeans}
   \label{alg-2}
\begin{algorithmic}
   \STATE {\bfseries Input:} $P=\{p_i, \mid 1\leq i\leq n\}\subset\mathbb{R}^d$, an integer $k\geq 1$, small $\epsilon>0$, and the integer lower and upper bounds $1\leq L\leq U\leq n$.
\begin{enumerate}
\item Run Chen's bi-criteria $k$-median/$k$-means algorithm~\cite{chen2009coresets}, and output the set $C$ of $O(k)$ cluster centers. 
\item Compute the $O(kn)$ distances from $P$ to $C$.
\item For each $k$-tuple $(s'_1, \cdots, s'_k)$ from $C^k$, solve the LP (\ref{for-lp1})-(\ref{for-lp3}).
\item Return the $k$-tuple having the smallest objective value and transform the LP solution to an integral solution via Theorem~\ref{lem-ext-integer}.
\end{enumerate}
\end{algorithmic}
\end{algorithm}

We present the algorithm for $k$-BMedian and $k$-BMeans in Algorithm~\ref{alg-2}. 
The time complexity of Step 1 and 2 of Algorithm~\ref{alg-2} is $O(knd)$. Using Theorem~\ref{lem-ext-integer}, we know that Step 3 and 4 takes $O(k^k)\times O(n+k^4 L^{3k}\log^2 L)=O(k^k n+k^{k+4}L^{3k}\log^2 L)$ time. 
%
%
%
%
For ease of presentation, we bound $r_{max}/r_{min}$ by the spread ratio $\Delta$, the ratio of the largest to smallest pairwise distance among $P$. Therefore, $L=O(\frac{\log \Delta}{\epsilon})$ (note that we do not need to really compute $\Delta$, since knowing $r_{max}$ and $r_{min}$ is already sufficient for Algorithm~\ref{alg-2}). Note that spread ratio is commonly used as a parameter of the time complexities in many geometric algorithms, and its logarithmic (i.e., $\log\Delta$) is usually not large~\cite{edelsbrunner1997cutting,indyk1999geometric}.

The resulting approximation factor of Algorithm~\ref{alg-2} depends on Lemma~\ref{lem-bridge1} and~\ref{lem-bridge2}. In addition, we slightly increase the approximation factor by a factor $1+\epsilon$ due to Lemma~\ref{lem-lp1}. Actually, we can set $\epsilon=1$ if only an $O(1)$-approximation is enough. In total, we have the following theorem.

\begin{theorem}
\label{the-ext1}
Algorithm~\ref{alg-2} yields an $O(1)$-approximation for $k$-BMedian/$k$-BMeans in $O\big(nd+(\log \Delta)^{3k}\log^2(\log\Delta)\big)$ time, if $k$ is a constant.
\end{theorem}

Since Chen's algorithm~\cite{chen2009coresets} and our Algorithm~\ref{alg-2} work for any metric space, we have the following corollary. 
\begin{corollary}
\label{cor-ext}
Given an instance of Metric $k$-BMedian/$k$-BMeans, we suppose the time complexity for acquiring the distance between any two vertices is $O(D)$. Algorithm~\ref{alg-2} yields an $O(1)$-approximation  in $O\big(nD+(\log \Delta)^{3k}\log^2(\log\Delta)\big)$ time, if $k$ is a constant.
\end{corollary}

\subsection{ $(1+\epsilon)$-Approximation Algorithms for $k$-BMedian and $k$-BMeans}
\label{sec-ext-ptas}
As mentioned in Section~\ref{sec-our}, we can use the ideas of~\cite{DBLP:conf/soda/DingX15,bhattacharya2018faster} to generate a set of candidates for the $k$ cluster centers; at least one $k$-tuple yields a $(1+\epsilon)$-approximation for $k$-BMedian/$k$-BMeans.


\begin{lemma}[\cite{bhattacharya2018faster}]
\label{lem-bridge3}
Let  $P$ be an instance of  $k$-BMedian (resp., $k$-BMeans). There exists an algorithm generating a set of $k$-tuples for $k$-BMedian (resp., $k$-BMeans) in $O(2^{\tilde{O}(k/\epsilon^{O(1)})}nd)$ (resp., $O(2^{\tilde{O}(k/\epsilon)}nd)$) time, such that with constant probability, at least one $k$-tuple yields a $(1+\epsilon)$-approximation. The size of the $k$-tuple set is $O(2^{\tilde{O}(k/\epsilon^{O(1)})})$ (resp., $O(2^{\tilde{O}(k/\epsilon)})$). The notation $\tilde{O}$ hides the logarithm factors $\log k$ and $\log (\frac{1}{\epsilon})$.
\end{lemma}


According to Lemma~\ref{lem-bridge3}, we can modify Step 1 and 2 of Algorithm~\ref{alg-2}, and improve the approximation factor to be $1+\epsilon$. 
First, for any of the $O(2^{\tilde{O}(k/\epsilon^{O(1)})})$ ($O(2^{\tilde{O}(k/\epsilon)})$) $k$-tuples, we need to compute the $kn$ pairwise distances to $P$. Second, the ratio $r_{max}/r_{min}$ is not bounded by the spread ratio $\Delta$. Actually, the algorithm of~\cite{bhattacharya2018faster} consists of a sequence of carefully designed sampling procedures, and each candidate cluster center is the mean of a multi-set sample of $O(\frac{k}{\epsilon^3})$ points from $P$. It implies that $r_{max}/r_{min}\leq O(\frac{k}{\epsilon^3}\Delta)$. 

\begin{theorem}
\label{the-ext2}
If we use the method of~\cite{bhattacharya2018faster} to generate the $k$-tuples as the candidates of cluster centers, Algorithm~\ref{alg-2} yields a $(1+\epsilon)$-approximation for $k$-BMedian (resp., $k$-BMeans) in $O\Big(2^{\tilde{O}(k/\epsilon^{O(1)})}\big(nd+ (\frac{\log (\Delta/\epsilon)}{\epsilon})^{3k}(\log\frac{1}{\epsilon}+\log\log\Delta)^2\big)\Big)$ (resp., $O\Big(2^{\tilde{O}(k/\epsilon)}\big(nd+ (\frac{\log (\Delta/\epsilon)}{\epsilon})^{3k}$ $(\log\frac{1}{\epsilon}+\log\log\Delta)^2\big)\Big)$)
time, if $k$ is a constant.
\end{theorem}

 \section{Acknowledgements} 
The author was supported by a start-up fund from Michigan State University and CCF-1656905 from NSF. Part of the work was done when the author was in IIIS, Tsinghua University and Simons Institute, UC Berkeley. The author also wants to thank Jian Li, Lingxiao Huang, Yu Liu, and Shi Li for their helpful discussion.

%
%
%
%
%
%
%
%
%
%
%
%
%
%
%
%
%
%


\bibliographystyle{abbrv}


\bibliography{BKC}

\section{Appendix}
\label{sec-other}
We address two questions following Lemma~\ref{lem-4app} for $k$-BCenter: (1) is the approximation factor $4$ tight enough, and (2) why should we use $S^k$ rather than $S$ directly?

For the first question, we consider the following example. Let $n=6$ points locate on a line, $k=3$, and $L=U=2$. See Figure~\ref{fig-4app}. It is easy to know that the optimal solution is $C_1=\{p_1, p_2\}$, $C_2=\{p_3, p_4\}$, and $C_3=\{p_5, p_6\}$ with $r_{opt}=1$. Suppose that the first point selected by Gonazlez's algorithm is $p_2$, then the induced $S=\{p_2, p_5, p_1\}$ which results in a $(4-\delta)$-approximation, no matter which $3$-tuple is chosen from $S^3$. Since $\delta$ can be arbitrarily small, the approximation ratio $4$ is tight.

\begin{figure}[h]
\vspace{-.2cm} 
\centering
 \includegraphics[height=1in]{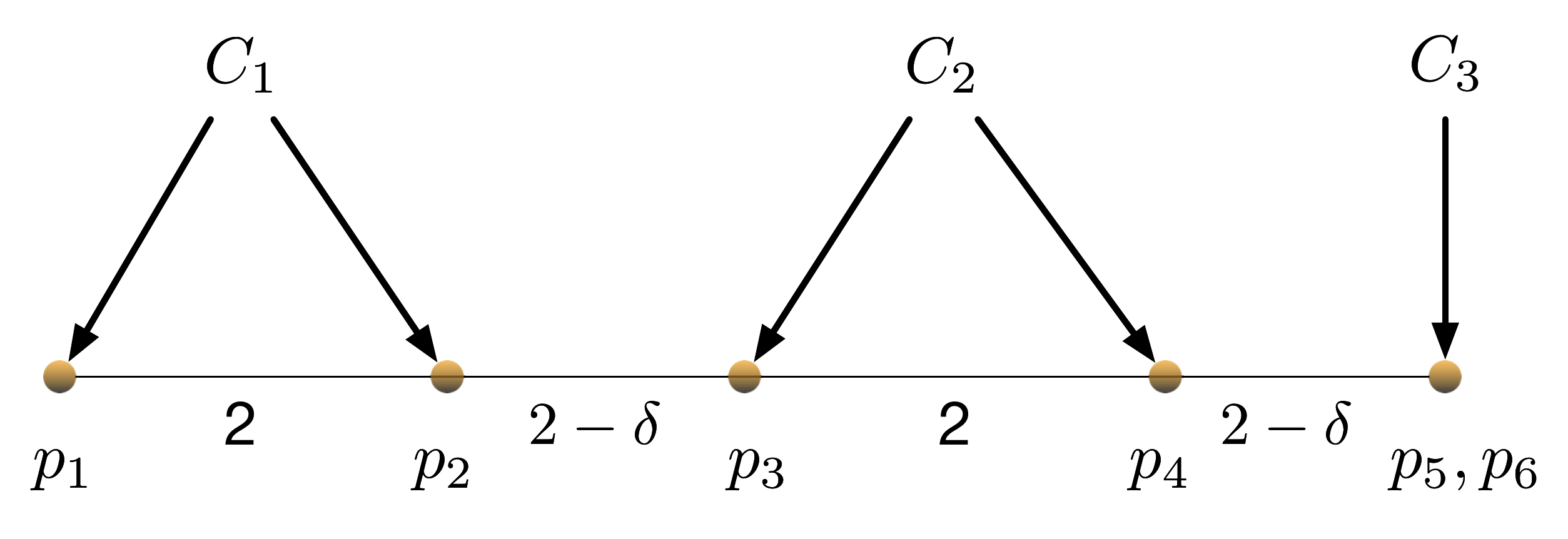}
  \vspace{-0.1in}
  \caption{$||p_1-p_2||=||p_3-p_4||=2$ and $||p_2-p_3||=||p_4-p_5||=2-\delta$ with a small positive $\delta$; $p_5$ and $p_6$ overlap.} 
   \label{fig-4app}
  \vspace{-0.15in}
\end{figure}

We construct another example to answer the second question. See Figure~\ref{fig-tuple}. It is easy to know $r_{opt}=r$. Suppose that the first point selected by Gonazlez's algorithm is $p_1$, then the induced $S=\{p_1, p_5, p_6\}$. If we take these 3 points as the cluster centers, the obtained radius is at least $h$ (since $p_3$ and $p_4$ have to be assigned to $p_6$). Consequently, the approximation ratio is $h/r$ that can be arbitrarily large. Hence we need to search the $k$-tuple points from $S^k$ rather than $S$. 

\begin{figure}[h]
\vspace{-.2cm} 
\centering
 \includegraphics[height=2in]{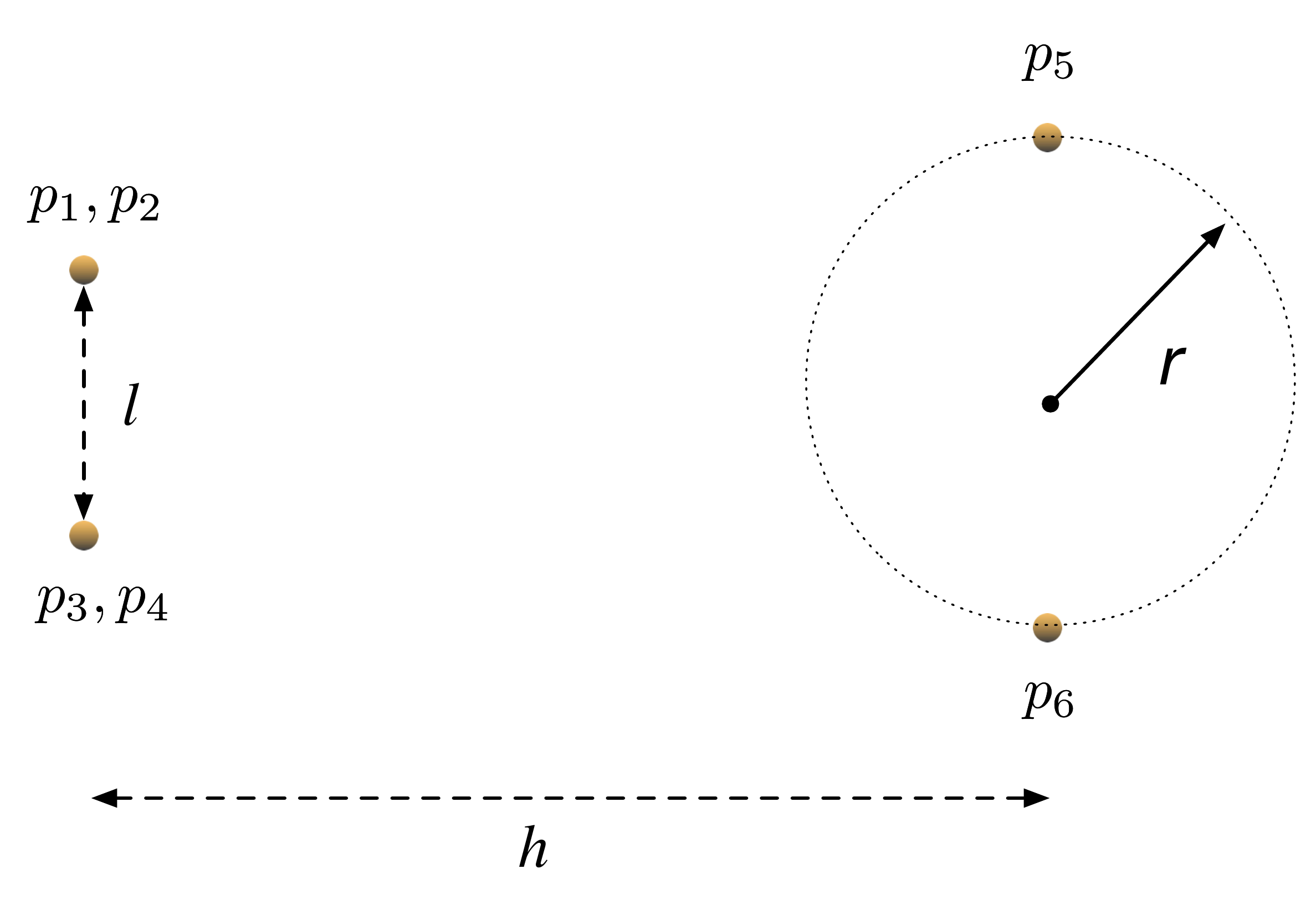}
  \vspace{-0.1in}
  \caption{Let the 6 points locate in a plane, $k=3$, and $L=U=2$. $p_1$ and $p_2$ overlap, $p_3$ and $p_4$ overlap, and these 4 points locate on the same vertical line while $p_5$ and $p_6$ locate on another vertical line; $||p_1-p_3||=l$, $||p_5-p_6||=2r$, and their horizontal distance is $h$; $l<2r\ll h$.} 
   \label{fig-tuple}
  \vspace{-0.15in}
\end{figure}


\end{document}